\documentclass[11pt]{article}

\usepackage[a4paper,margin=1in]{geometry}
\usepackage{setspace}
\usepackage{authblk}

\usepackage{amsmath, amssymb, amsfonts}
\usepackage{amsthm}
\usepackage{mathtools}
\usepackage{bbm}
\usepackage{tikz}
\usepackage{amsthm}       
\usepackage{hyperref}     

\theoremstyle{plain}
\newtheorem{theorem}{Theorem}[section]
\newtheorem{lemma}[theorem]{Lemma}

\theoremstyle{definition}

\theoremstyle{remark}
\newtheorem{remark}[theorem]{Remark}
\newtheorem{corollary}[theorem]{Corollary}

\DeclareMathOperator{\Dom}{Dom}

\newcommand{\C}{\mathbb{C}}
\newcommand{\R}{\mathbb{R}}
\newcommand{\N}{\mathbb{N}}
\newcommand{\Z}{\mathbb{Z}}

\newcommand{\ket}[1]{\lvert #1 \rangle}
\newcommand{\bra}[1]{\langle #1 \rvert}

\usepackage{graphicx}
\usepackage{float}

\usepackage[numbers]{natbib}
\usepackage{hyperref}
\hypersetup{
    colorlinks=true,
    linkcolor=blue,
    citecolor=blue,
    urlcolor=blue
}

\usepackage{enumitem}
\setlist{nosep}
\usepackage{cleveref}

\title{Stark localization of interacting particles}
\author[1]{Wojciech De Roeck}
\author[2]{Amirali Hannani}
\author[3] {Alessio Lerose}
\author[4]{Nathan Vandenbosch}
 \affil[1,3,4]{Institute for Theoretical Physics, 
 KU Leuven, 3000 Leuven, Belgium}
\affil[2]{Section of Mathematics, University of Geneva, 1205 Genève, Switzerland}
\date{\today}

\begin{document}

\maketitle

\begin{abstract}
    We consider $N$ interacting quantum particles on a one-dimensional lattice, and  subjected to an external linear potential. For $N=1$, the corresponding Hamiltonian  is explicitly diagonalizable, with superexponentially localized eigenstates. This is called Stark localization.
    We prove that superexponential spectral localization persists for arbitrary $N$ and every interaction strength.
\end{abstract}

\section{Introduction}

The phenomenon of Anderson localization \cite{Anderson, abrahams1979scaling,gol1977pure, kunz1980spectre, Frohlich1983,  Aizenman1993} has drawn a lot of interest in the past decades, both from the physics and the mathematics community.
Whereas the original problem was concerned with the behavior of a single particle, the focus has shifted in recent years to few-body localization  \cite{Chulaevsky2011, Chulaevsky2014, Aizenman2009, Fauser2014}  and to many-body localization \cite{Basko2006,gornyi2005interacting,oganesyan2007localization}. Here, by the few-body setup we understand that the number of particles is fixed, independently of the volume, whereas the many-body setup refers to a \emph{positive density}, where the number of particles is proportional to the total volume. 

The present paper is concerned with \emph{Stark} few-body localization in dimension $1$.
Stark localization differs from Anderson localization in that the potential  $V(x), x\in \Z$ acting on individual particles takes the form of a deterministic linear gradient $V(x)=hx$, with $h\neq 0$, rather than $(V(x))_{x\in \Z}$ being a collection of random variables, as it is for Anderson localization. For a single particle, Stark localization is an exactly solvable problem: the Hamiltonian can be explicitly diagonalized, and the spatial localization of eigenfunctions directly follows from energy conservation~\cite{wannier_dynamics_1962}.  That is no longer true for the problem with which we are concerned here: interacting particles in a Stark potential. 


The possibility of a Stark few-body and many-body localization phenomenon was proposed in~\cite{van_nieuwenburg_bloch_2019,schulz_stark_2019} and  elucidated in~\cite{taylor_experimental_2020,kohlert_exploring_2023,krajewski_dependence_2025,bertoni_local_2024,doggen_stark_2021,zisling_transport_2022, SunWang2025Stark},  along with experimental observations of signatures of this phenomenon~\cite{morong_observation_2021,wang_exploring_2025,guo_stark_2021}.
Already in early investigations it was argued that the resonances associated with a perfectly linear potential may enable a breakdown of localization visible on very long time-scales, and it was proposed that a weak non-linear curvature~\cite{schulz_stark_2019} or weak disorder~\cite{van_nieuwenburg_bloch_2019} restore many-body localization.

The central theoretical question we address in this work is: in what sense and under which conditions do Stark few-body phenomena exist?  
We provide a partial answer to this question:
We prove that an arbitrary but finite number of particles in infinite volume are spectrally localized, i.e.\ the spectrum is pure point. 
\\
We emphasize that our result differs from the existing mathematical literature both technically and conceptually. Rigorous mathematical results about Stark localization mostly study the spectral properties of the one body problem perturbed by certain potentials, either in continuum  
\cite{herbst1981stark,frank2022spectrum,stolz1993note,hislop1990stark,korotyaev2020eigenvalue,korotyaev2017resonances,sacchetti2014existence,grecchi1994stark,bentosela1991stark,combes1991stark},
or on the lattice, \cite{sun2025stark,aloisio2025dynamical,longhi2023absence}. Our result, however, is concerned with the \textit{few body} setup. 


\section{Result }

The one particle Stark Hamiltonian describes a 
particle hopping on a one-dimensional lattice and subjected to a linear field. It is given by  
$$
H_0=g\Delta -2 h X,
$$
acting on a dense subspace of $\ell^2(\Z)$, where
\begin{enumerate}
    \item $X|x\rangle=x|x\rangle$ with $|x\rangle$ the canonical basis vectors in $\ell^2(\Z)$. \item $\Delta$ is the lattice Laplacian. 
    \item $g,h$ are real parameters. 
\end{enumerate}
By the Kato-Rellich theorem,  $H_0$ is self-adjoint on the domain of $X$. It is well-known that, for any nonzero value of $h$, it has pure point spectrum, and the eigenvectors are strongly localized in space (they decay faster than exponential away from a localization center).    This phenomenon is called ``Stark localization".  

We consider the case of  $N$ particles. Define on (a dense subspace of) the $N$-particle space $\bigotimes_{j=1}^N \ell^2(\Z) = \ell^2 (\Z^N)$, with inner product $\langle\cdot | \cdot\rangle$, the operator
$$
H_0^{(N)} = H_0 \otimes 1 \otimes \dots \otimes 1 + 1 \otimes H_0 \otimes\dots \otimes 1 + \dots + 1 \otimes 1 \otimes \dots \otimes H_0.
$$
To get the full Hamiltonian, we add a pair interaction. 
$$
H^{(N)} =H^{(N)}_0 +   \sum_{1\leq i<j\leq N}V({x_i-x_j})
$$
where $V({x_i-x_j})$ acts by multiplication with the function $v({x_i-x_j})$  in the position basis $|x_i,x_j\rangle$.  We assume that the pair potential $v:\Z\to\R$ is bounded and decays at infinity.

From a physics perspective, it is more natural to consider fermionic particles. This is possible and the results for fermions (or bosons) follow straightforwardly from our results for distinguishable particles, see subsection \ref{sec: fermions}.

As a general (non-rigorous) rule in many-body physics, interaction tends to destroy localization phenomena.  
However, in the above setup we find that this is not the case. At least the spectral localization is \emph{not} affected by the interaction.

\begin{theorem}[Spectral Localization]\label{thm: spectral localization}
 If $h\neq 0$,   the operator $H^{(N)}$ has pure point spectrum. 
\end{theorem}

Let us comment on what this implies for the dynamical behaviour. 
One might consider the local particle density $\rho(x,t)$ corresponding to a normalized initial state $\psi_0 \in \ell^2(\mathbb{Z}^N) $;
$$
\rho(x,t)= \sum_{i=1}^N \sum_{\underline{x}\in \Z^N:}  \delta_{x_i,x}|\psi_t(\underline{x})|^2, \qquad \psi_t= e^{-it H^{(N)}}\psi_0
$$
normalized to $\sum_x \rho(x,t)=N$, where $\delta_{x_i,x}$ is the Kronecker delta. By expanding $\psi_t$ in the basis of eigenvectors, our result immediately implies  
\begin{corollary}
Let us fix an arbitrary $N$ and an initial normalized state $\psi_0$. Then
$$
\lim_{r\to\infty}\sup_{t\in \R} \sum_{x: |x| > r}  \rho(x,t) =0.
$$
\end{corollary}
That is, the wavefunction does not escape to infinity. \\

This result does however not settle the physically relevant question of \emph{dynamical localization}. We leave this for future work and we refer to \cite{PhysRevLett.75.117,del1996operators} for an insightful discussion of the possibility having spectral localization without dynamical localization. 
In the physics literature, quite some recent discussion on Stark localization was focused on transport properties, hence on dynamical localization~\cite{sala_ergodicity_2020,khemani_localization_2020,burin_exact_2022,nandy_emergent_2024,ribeiro_many-body_2020}.
Since we have not established dynamical localization, we cannot 
rule out the very slow transport of particles that is expected in some of the above works. \\

We can quantify the localization properties by exhibiting the exponential decay of some eigenvectors.
For every $N$, let $\mathcal P_N$ be the set of partitions $P=\{p_1,\ldots,p_{|P|}\}$ of $N$, i.e.\ a collection of $|P|$ nonzero naturals such that $\sum_{i=1}^{|P|}p_i=N$. 
We define then
$$
\sigma^{(N)}_{\mathrm{cluster}} =\bigcup_{P \in \mathcal P_N}\left(\sigma(H^{(p_1)})+\ldots +  \sigma(H^{(p_{|P|})})\right)
$$
where $\sigma(A)$ is the spectrum of an operator $A$, by $H^{(p_i)}$ we mean the Hamiltonian with $p_i$ particles, and the sums are Minkowski sums. 
\begin{theorem}[Superexponential localization]\label{thm: superexp loc}
    Let $N \geq 1$, and $\theta \geq 0$. Assume $\lambda$ is an eigenvalue of $H^{(N)}$ such that $\lambda \notin \sigma^{(N)}_{\mathrm{cluster}}$.  
 Then, $\lambda$ has finite multiplicity and  there exists a constant $C(N,\lambda, \theta)$  such that, for any corresponding normalized eigenvector $\psi_{\lambda}$,
    $$
    |\langle \psi_{\lambda}\mid x_1,x_2,\dots ,x_N\rangle | <C(N,\theta,\lambda)  e^{-\theta (|x_1| + |x_2| + \dots + |x_N|)}.
    $$
\end{theorem}
Three remarks are in order:\\

\noindent\textbf{Remark 1.} We expect that, as one varies $g,h$ and the potential $v$, only isolated values give rise to spectrum in $\sigma^{(N)}_{\mathrm{cluster}}$, and so, for a typical choice of parameters, the above theorem would in fact apply to \emph{all} eigenvectors. However, we have not been able to prove this.  \\

\noindent\textbf{Remark 2.} The spatial decay is with respect to the origin. One might find it more natural to find a localization center and express the spatial decay relative to the localization center. This can be trivially achieved since the constant $C(N, \lambda,\theta)$ depends on the eigenvalue $\lambda$. However, one should contrast the situation with the case of Anderson localization of particles in a disordered potential. In that case, and restricting to the one-particle problem,  one can prove an estimate of the form 
$$|\langle \psi_{\lambda}|x\rangle| \leq C (1+x_{\textrm{loc}}^2) e^{-\theta |x-x_{\textrm{loc}}|}, $$
where  $\theta>0$  is now not longer arbitrary, and  
with $C$ a random variable with finite expectation value. That is, the prefactor grows only mildly with the localization center.  At the time of writing, we cannot prove a bound of this type.  Stated informally, we cannot rule out the case where the eigenfunctions have an anisotropic decay away from their localization center: faster than exponential towards infinity, and slower than exponential towards the origin. \\

\noindent\textbf{Remark 3.}  
Although we cannot prove the result mentioned in the above remark, we are still able to prove a result in this direction. In corollary \ref{cor: sum decay}, we prove that for any $\theta>0$ there is $C(\theta)$ such that and any eigenvalue/eigenvector pair $\lambda, \psi_{\lambda}$ we have: 
$$|\langle \psi_{\lambda}\mid x_1,\dots x_n\rangle| \leq C(\theta) \exp(-\theta|\sum_i x_i -\lambda|).$$
This result is valid for all $\lambda$, and the constant is uniform in $\lambda$, unlike the above-mentioned theorem. However, it only provides decay in "one direction": the direction of the center of mass. 
\\


\subsection{Fermions and bosons}\label{sec: fermions}

Let $P^{{(N)}}_{\mathrm{sym}}$ be the orthogonal projector onto the fully (anti)symmetric subspace of $\ell^2(\Z^{N})$, i.e.\ the space consisting of $\psi \in \ell^2(\Z^{N})$ such that
\begin{equation}\label{eq: antisym}
     \psi(x_{\pi(1)},\ldots,x_{\pi(N)})=\eta^{\mathrm{sign}(\pi)} \psi(x_1,\ldots,x_N), \qquad \eta=\pm 1
\end{equation}
for every permutation $\pi$ of the index set $\{1,\ldots,N\}$, and with $\eta=1, -1$ corresponding to bosons or fermions, respectively. 
If the pair potential $v:\Z \to \R$ is chosen symmetric, i.e.\ $v(x)=v(-x)$, then we check by inspection that
$$
[P^{{(N)}}_{\mathrm{sym}},H^{(N)}]=0.
$$
Since $H^{(N)}$ can be viewed as a direct sum of operators that are unitarily equivalent to $P^{{(N)}}_{\mathrm{sym}} H^{(N)} P^{{(N)}}_{\mathrm{sym}}$, it follows that spectrum and spectral type of  $H^{(N)}$ and $P^{{(N)}}_{\mathrm{sym}} H^{(N)} P^{{(N)}}_{\mathrm{sym}}$ coincide, only the multiplicity differs. 

Therefore, Theorems \ref{thm: spectral localization} and Theorem \ref{thm: superexp loc} carry over without any changes to the case of the fermionic or bosonic Hamiltonian. That is, one can replace $H^{(N)}$ by  $P^{{(N)}}_{\mathrm{sym}} H^{(N)} P^{{(N)}}_{\mathrm{sym}}$.

\section{Notation and preliminaries}
To avoid confusion with the to-be-introduced basis of Stark eigenstates, we write $e_j$ for the canonical basis of $\ell^2(\Z) $, such that the operators $X$ and $\Delta$, introduced above, are defined as
$$
Xe_j =  je_j,  \qquad  \Delta e_j = - (e_{j-1} +e_{j+1}).
$$
As above, we write $H_0=-2hX+g\Delta$ for the one-particle Hamiltonian, self-adjoint on the domain of $X$. 

It is known that, for any $g \geq 0$, $ H_0 $ is diagonalizable \cite{Fukuyama1973}. Its eigenvalues are $\{-2h m \mid m \in \Z\}$ with corresponding orthonormal basis of eigenvectors $\xi^{(m)} \in \ell^{2}(\mathbb{Z})$ given by $\xi_j^{(m)} :=\langle e_j \mid \xi^{(m)} \rangle= \mathcal{J}_{m-j}(g/h)$. Here $\mathcal{J}_n$ denotes the $n$-th Bessel function. We will throughout this paper denote $\xi^{(m)}$ by $\ket{m}$, and when going to higher (distinguishable) particle number write $\ket{m_1,m_2, \dots, m_N}$ to mean $\xi^{(m_1)} \otimes \xi^{(m_2)} \otimes \dots \otimes \xi^{(m_N)}\in\bigotimes_{k = 1}^N \ell^2(\Z) = \ell^2(\Z^N)$. We use the shorthand notation $\underline{m} := (m_1,\dots, m_N)$ for any $\underline{m} \in \mathbb{Z}^N$. We also represent $\ket{m_1,\dots, m_N}$ by $\ket{\underline{m}}$ when it is convenient. \\
Let $v: \Z \to \R$ be a bounded function, such that $\lim_{n \to \pm\infty} v(n) = 0$. We define the two-body interaction
$$
V^{(2) } (e_{j_1} \otimes e_{j_2}) = v(j_1-j_2) (e_{j_1} \otimes e_{j_2}).
$$
In the $\ket{m_1,m_2}$ basis this is given by
$$
V^{(2)}_{n_1,n_2;m_1,m_2}:=\bra{n_1,n_2} V^{(2)} \ket{ m_1,m_2} = \sum_{j_1,j_2}v(j_1-j_2) \mathcal{J}_{n_1-j_1}\mathcal{J}_{m_1 -j_1} \mathcal{J}_{n_2-j_2}\mathcal{J}_{m_2 -j_2}
$$
where we have dropped the $g/h$ dependence in the Bessel functions, for readability.
 Fix $N \geq 1$ to denote the number of particles. We define  the operator
$$
H_0^{(N)} = H_0 \otimes 1 \otimes \dots \otimes 1 + 1 \otimes H_0 \otimes\dots \otimes 1 + \dots + 1 \otimes 1 \otimes \dots \otimes H_0, 
$$ on a proper dense subset of $\bigotimes_{j} \ell^2(\Z) = \ell^2 (\Z^N)$ such that its closure is self-adjoint. 
Hence $$H_0^{(N)}\ket{\underline{m}} = -2h(m_1+\dots+m_N) \ket{\underline{m}}.$$ 
Let $\Lambda^{(N)}$ denote the $\binom{N}{2}$ unordered pairs $\{(i,j) \in \{1,\dots N\}^2 \mid  i<j\} .$ For any $\alpha \in \Lambda^{(N)}$, denote with $V_{\alpha}\equiv V_\alpha^{(N)}$ the operator on $\ell^2(\Z^N)$ given by
$$
V_\alpha^{(N)} \ket{m_1, \dots m_i, \dots  ,m_j , \dots , m_N} = \sum_{n_i,n_j \in \Z^2} V_{n_i,n_j;m_i,m_j}^{(2)} \ket{ m_1,\dots n_i, \dots , n_j, \dots, m_N}.
$$
Let 
$$
V \equiv V^{(N)} := \sum_{\alpha \in \Lambda^{(N)}} V_{\alpha}^{(N)}
$$
and let 
\begin{equation} \label{eq: Ham}
H^{(N)} = H_0^{(N)} + V ^{(N)}= H_0^{(N)}+V.
\end{equation}

\subsection{Conventions}
We use the following conventions throughout the paper: 
\begin{enumerate}
    \item Generally, we denote the closure of any densely defined (possibly unbounded) operator with the same symbol as the operator. Occasionally, we use the over-bar notation $\over{\cdots}$ to highlight the distinction. 
    \item By abusing the notation, for $\alpha=(i,j)$, we sometimes represent $V^{(2)}_{n_i,n_j;m_i,m_j}$ by $\langle n_1,n_2\mid V_{\alpha} \mid m_1,m_2\rangle\equiv \langle n_1,n_2\mid V_{\alpha}^{(N)} \mid m_1,m_2\rangle$.
    \item By the symbol $n!$, in case $n \notin\mathbb{Z}$ we mean $\Gamma(n+1)$ whenever $\mathrm{Re}(n)>0$. 
    \item We drop the superscript $N$ when there is no confusion. We mostly keep superscript of $H^{(N)}$  due to inductive structure as already depicted in the definition of $\sigma_{\mathrm{cluster}}^{(N)}$.
\end{enumerate}

\subsection{Basic properties of Bessel functions}
Here we gather some basic properties of Bessel functions: 
\begin{lemma} \label{lem: simple bounds}
    Let $\mathcal{J}_n(x)$ denote the Bessel function. Then we have for any $x \in \mathbb{R}$: 
    \begin{enumerate}
        \item Simple upper bound: 
        \begin{equation} \label{eq: Bessel UB}
           |\mathcal{J}_n(x)| \leq \frac{1}{|n|!} \left(\frac{|x|}{2}\right)^{|n|}
       \end{equation}  \item Summability: 
        \begin{equation} \label{eq: Bessel Sum}
            \sum_{n \in \mathbb{Z}} |\mathcal{J}_n(x)|\leq C \equiv C(|x|) 
        \end{equation}
        \item Decay property:
        \begin{equation} \label{eq: Bessel decay 1}
            \sum_{j \in \mathbb{Z}} |\mathcal{J}_{n-j}(x) \mathcal{J}_{m-j}(x)| \leq C\frac{\exp(c|m-n|)}{|(m-n)/2|!},
        \end{equation}
        where $C,c\geq 0$ depend on $|x|$. 
        \item There exists a function $f$ such that $\lim_{n \to \pm \infty} f(n)=0$ and we have for any $m_1,m_2 \in \mathbb{Z}$: 
        \begin{equation} \label{eq: Bessel decay with V}
            \sum_{j_1,j_2 \in \mathbb{Z}} |v(j_1-j_2)||\mathcal{J}_{m_1-j_1}(x) \mathcal{J}_{m_2-j_2}(x)| = f(m_1-m_2)
        \end{equation}
    \end{enumerate}
    
\end{lemma}
\begin{proof}
    1) See \cite[Pg. 49, Equation (1)]{watson1944treatise} for $n \geq 0$. Since $\mathcal{J}_{-n}(x) = (-1)^n\mathcal{J}_n(x)$ for integers, this extends to $n\leq 0$ as well. \\
2) By \cref{eq: Bessel UB} we see that 
$$
\sum_{n\in \Z}|\mathcal{J}_n(x)|  \leq \sum_{n \in \Z} \frac{1}{|n|!} \left(\frac{|x|}{2}\right)^{|n|} = 2 \exp({|x|}/2) - 1
$$
by the definition of the exponential.\\
3) We use that for $a,b \in \mathbb{N} $ we have that 
${1}/{|a+b|! |a-b|!} \leq {1}/{|a|! |b|!}$. By simple inspection this can be extended to all $a,b \in \frac{1}{2}\Z$. In particular, in our case for $a = {(m+n-2j)}/{2}$ and $b = {(m-n)}/{2}$, by \cref{eq: Bessel UB} we get
\begin{align*}
\sum_{j \in \mathbb{Z}} |\mathcal{J}_{n-j}(x) \mathcal{J}_{m-j}(x)| &\leq \sum_{j \in \mathbb{Z}}\frac{|x|^{|n-j|}}{|n-j|!} \frac{|x|^{|m-j|}}{|m-j|!}  \\
&\leq \sum_{j \in \mathbb{Z}}\frac{1}{|{(m+n-2j)}/{2}|!} \frac{1}{| {(m-n})/{2}|!} {|x|}^{|m-j|+|n-j|}.
\end{align*}
Taking $c= \max(\ln(|x|),1)$ and following observations finishes the proof: 
first:
$$
|n-m| \leq |n-j| + |m-j| \leq |n-m| + |n+ m - 2j|
$$
and second, by definition of the $\exp$  function:
$$ \sum_{j\in \mathbb{Z}} \frac{|x|^{|m+n-2j|}}{|(m+n-2j)/2|!} \leq  C\exp(c|x|^2).$$

4) By shifting summation indexes, it is clear that left hand side of \eqref{eq: Bessel decay with V} is only dependent on $m_1-m_2$. Hence define
$$
            f(n) := \sum_{j_1,j_2 \in \mathbb{Z}} |v(j_1-j_2)| |\mathcal{J}_{-j_1}(x) \mathcal{J}_{n-j_2}(x)|.
$$
We now prove that $\lim_{n\to \pm \infty} f(n) = 0$. Indeed, we get that 
\begin{align*}
\lim_{n\to \pm \infty}\sum_{j_1,j_2 \in \mathbb{Z}} |v(j_1-j_2)||\mathcal{J}_{-j_1}(x) \mathcal{J}_{n-j_2}(x)| |&= \lim_{n\to \pm \infty}\sum_{j_1,j_2 \in \mathbb{Z}} |v(j_1-j_2-n)||\mathcal{J}_{-j_1}(x) \mathcal{J}_{-j_2}(x)| \\
&\leq \sum_{j_1,j_2 \in \mathbb{Z}} \lim_{n\to \pm \infty}|v(j_1-j_2-n)||\mathcal{J}_{-j_1}(x) \mathcal{J}_{-j_2}(x)|\\
&= 0
\end{align*}
by the dominated convergence theorem, since we know that $|v(j_1-j_2-n)||\mathcal{J}_{-j_1}(x) \mathcal{J}_{-j_2}(x)|$ is dominated by $ \| v\|_{\infty}  |\mathcal{J}_{-j_1}(x) \mathcal{J}_{-j_2}(x)|$ for all $n \in \Z$, which is a summable function by \cref{eq: Bessel Sum}.
\end{proof}
\begin{remark}
    Let us remark that the above Lemma in principle means that we have the following asymptotic decay for matrix elements of $V^{(2)}$: 
     $$
    |\bra{n_1,n_2} V^{(2)} \ket{ m_1,m_2}| \lesssim f((m_1 - m_2) + (n_1-n_2))  \frac{\exp({c(|m_1 -n_1|+|n_2-m_2|))}}{|{|n_1-m_1|}/{2}|!\,|{|n_2-m_2|}/{2}|!},
    $$
    with $C,c \geq 0$ are constants depending on $g/h$ and $f$ decays to zero at $\pm \infty $. More precisely, division of $|m_i-n_i|$ by  two can be omitted as well. We dropped this extra precision, since it does not bring any new insight and just makes the proofs longer.
\end{remark}

\section{$H^{(N)}$ has pure point spectrum} \label{sec: pp spectrum}
 The goal of this section is to show that $H^{(N)}$ has pure point spectrum. We will do that using techniques developed in \cite{HunzikerWalter1966Otso}. The main tool that we use is equation \eqref{functional_equation} which also appeared in \cite{HunzikerWalter1966Otso}. Let us give a general picture about this equation: \eqref{functional_equation} states that $G(z)$ (resolvent of $H^{(N)}$) satisfies a functional equality of the form $G(z)=D(z)+I(z)G(z)$ for given operators $I$ and $D$. The main point is that we can show that $I$ is a compact operator on a proper domain and both $I$ and $D$ are well defined on the complement of $\sigma_{\mathrm{cluster}}^{(N)}$. Once these properties are established, deducing that the spectrum is countable is direct. Therefore, the main task of this section is to construct $I$ and $D$ and verify the mentioned properties. The construction is inspired by the cluster expansion scheme of \cite{HunzikerWalter1966Otso}, which is done in an inductive manner.  Our main contribution is to verify the mentioned properties, using specific form of our Hamiltonian.    \\

The derivation of \eqref{functional_equation} is similar to the derivation in \cite[Section 5]{HunzikerWalter1966Otso}, which we repeat here, since the machinery of this derivation will be used extensively. Let $D_k = \{C_i\}_{1 \leq i \leq k}$ be a partition of $\{1,\dots, N\}$ into $k$ subsets. We will call $C_i$ a cluster, and $D_k$ a cluster-decomposition. For such a cluster-partition, define 
$$
H_{D_k}=H_{D_k}^{(N)} = H_0^{(N)}+ \sum_{i=1}^k\sum_{\alpha \in C_i^{2}} V_\alpha,
$$
where the second sum is only over ordered pairs $\alpha \in C_i^2$ (i.e. if $\alpha=(i,j)$ then $i<j$). This is the Hamiltonian which only allows interactions within each cluster $C_i$.\\
Denote with $G(z)=G^{(N)}(z) = (z-H^{(N)})^{-1}$ the resolvent of $H^{(N)}$, and denote with $G_0(z)=G^{(N)}_0(z) = (z-H_0^{(N)})^{-1}$ the resolvent of $H_0^{(N)}$.  \\
Since $\sigma(H_0^{(N)})=2h\mathbb{Z}$, $\|G_0(z)\| \to 0$ when $|\operatorname{Im}(z)| \to \infty$. Hence, since $V$ is bounded, there exists an $M>0$ such that when $|\operatorname{Im}(z)| > M$,  we have that $\|G_0(z) V\| < 1$, hence that the series
\begin{align*}
G(z) &= \sum_{n=0}^{\infty} (G_0(z) V)^n G_0(z)  \\
&= \sum_{n=0}^{\infty} \sum_{ \alpha_1 \in \Lambda^{(N)}} \dots \sum_{ \alpha_n \in \Lambda^{(N)}}G_0(z) V_{\alpha_1} G_0(z) \dots G_0(z)V_{\alpha_n}G_{0}(z)
\end{align*}
converges absolutely in norm.\\
We will now classify the terms $G_0(z) V_{\alpha_1} G_0(z) \dots G_0(z)V_{\alpha_n}G_{0}(z)$ appearing inside the sum, or in other words classify $(\alpha_1, \alpha_2, \dots, \alpha_n)\in (\Lambda^{(N)})^n$. We do this by constructing a corresponding graph in the following manner: Let $(\alpha_1, \alpha_2, \dots, \alpha_n)\in (\Lambda^{(N)})^n$ be a sequence. First draw $N$ horizontal lines, representing $N$ particles. Draw $n$ vertical lines, one for every $\alpha_i =( \alpha_{i,1}, \alpha_{i,2})$, drawn in order from left to right connecting the $\alpha_{i,1}$-th and $\alpha_{i,2}$-th horizontal line.\\
For example, the following graph would represent the term $G_0(z) V_{(2,N)} G_0 V_{(1,2)} G_0 V_{(1,2)} G_0 V_{(1,3)} G_0(z)$. \\
\begin{figure} [hbt!]
    \centering
\begin{tikzpicture}

\draw (0,2.7) -- (12,2.7);
\draw (0,1.8) -- (12,1.8);
\draw (0,0.9) -- (12,0.9);
\draw (0,-0.9) -- (12,-0.9);

\draw[dashed] (6,-1.35) -- (6,3.15) node[above] {$L$};

\draw (1.5,1.8) -- (1.5,-0.9);
\draw (4.5,2.7) -- (4.5,1.8);
\draw (7.5,2.7) -- (7.5,1.8);
\draw (10.5,2.7) -- (10.5,0.9);

\foreach \x/\y in {
  1.5/1.8, 1.5/-0.9,
  4.5/2.7, 4.5/1.8,
  7.5/2.7, 7.5/1.8,
  10.5/2.7, 10.5/0.9
}
  \draw (\x,\y) circle (2pt);

\node[left] at (0,2.7) {1};
\node[left] at (0,1.8) {2};
\node[left] at (0,0.9) {3};
\node[left] at (0,-0.9) {$N$};

\node[left] at (0,0) {$\vdots$};

\end{tikzpicture}

    \label{fig:placeholder}
\end{figure}

Now note the following
\begin{enumerate}
    \item For each graph $K = (\alpha_1,\alpha_2,\dots, \alpha_n)$, one defines a corresponding cluster decomposition $D(K)$ as the connected parts of $K$: $i$ and $j$ are in the same cluster if and only if there is a path linking them, using only the edges of $K$. For the above graph, the corresponding cluster decomposition is $(1,2,3,N)(4),\dots,(N-1) $.
    \item Let $D_k = \{C_i\}_{1 \leq i \leq k} $ and $ D_l' = \{C_j'\}_{1 \leq j \leq l}$ be two cluster partitions. We say $D_l' \subset D_k$ if every cluster of $D_k$ is contained in some cluster of $D_l'$, and at least one is strictly smaller. 
    \item We say that a graph $K$ is $D_k$-disconnected if $D_k \subseteq D(K)$. Notice that unlike the previous definition, here we allow for the equality. Note that  
$$
\sum (\text{all $D_l$-disconnected graphs}) = (z-H_{D_l})^{-1} =\colon G_{D_l}(z).
$$
  where in the above expression by a graph $(\alpha_1,\dots,\alpha_n)$ we mean the corresponding term $G_0V_{\alpha_1}G_0\dots V_{\alpha_n}G_0$. In fact, throughout this section, we use the word graph for both of these purposes. The meaning should be clear from the context. 
    \item  For a graph $K = (\alpha_1,\alpha_2,\dots , \alpha_n)$ we can obtain a sequence of cluster decompositions $S = (D_N, D_{N-1}, \dots, D_{k})$: Draw a vertical line $L$ (see figure), and let $D(L)$ denote the decomposition of the subgraph to the left of $L$. To find $D_N$, let $L$ be to the left of all interactions, then $D(L) = (1)(2)\dots (N)$. Shifting $L$ to the right, we obtain a sequence of cluster decompositions 
    $S = (D_N, D_{N-1}, \dots D_{k})$ such that $D_{i+1} \supset D_i$ (notice the strict inclusion). In our example, the sequence would be $$S = ((1)(2)\dots(N), (1)(3)\dots(2,N),(3) \dots  (1,2,N), (1,2,3,N)\dots).$$
    Every graph $K$ defines a sequence of cluster decompositions $S(K)$. Conversely, let $S = (D_N, D_{N-1}, \dots D_k)$ be a sequence of cluster decompositions such that $D_{i+1}\supset D_i$. We define the class of $S$ as being all graphs $K$ such that $S(K) = S$. To each $S$ we associate $k_S$ such that length of $S$ is equal to $N-k_S+1$. Note that $k_{S(K)} = |D(K)|$.
\end{enumerate}
Next, fix a sequence of cluster decompositions $S = (D_N,D_{N-1}, \dots,D_k)$. 
For any $i$, we define: $$\mathcal{C}_i=\mathcal{C}(D_{i+1},D_{i}):= \{\alpha| \alpha \text{ belongs to a cluster of } D_i \text{ and it does not belong to any cluster of } D_{i+1}\}.$$
Note that there is a bijection between $\{ \text{Graphs of class } S\}$ and 
$$
\bigtimes_{i= N-1}^{k} (\mathcal{C}(D_{i+1},D_{i}) \times \{ D_i\text{-disconnected graphs}\})
$$
Hence if we define $G_S(z) \coloneq \sum (\text{all graphs of class } S)$ we get that
\begin{align}
    G_S(z) &=  \sum (\text{all graphs of class } S)\\
    &= G_0(z) \prod_{i = N-1}^k (\sum _{\substack{\alpha \in \\ \mathcal{C}(D_{i+1},D_{i})}} V_{\alpha})(\sum_{\substack{(\alpha_1,\alpha_2,\dots, \alpha_l)  \text{ is }\\ D_i\text{-disconnected}}} G_0(z) V_{\alpha_1}G_0(z) V_{\alpha_2} G_0(z) \dots G_0(z) V_{\alpha_l}G_0(z))\\
    &=G_0(z) \prod_{i = N-1}^k V_{D_{i+1},D_i} G_{D_i}(z) \\
&= G_{D_N} V_{D_N,D_{N-1}}G_{D_{N-1}} V_{D_{N-1},D_{N-2}} G_{D_{N-2}}\dots V_{D_{k+1},D_k}G_{D_k},\label{thesum}
\end{align}
where we defined $V_{D_{i+1},D_i} = \sum _{\alpha \in \mathcal{C}( D_{i+1},D_i)} V_{\alpha}$.\\
Define now
$$
D(z) = \sum_{ \substack{\text{all } S \text{ with } \\k_S\geq 2}} G_S(z)
$$
and 
$$
C(z) = \sum_{\substack{ \text{all } S \text{ with }\\ k_S=1}} G_S(z)
$$
Note that when $S$ has $k_S =1$, then the last occurrence in the sum \cref{thesum} is always $G(z)$. Factoring this out, and defining 
$$
I(z) \coloneq \sum_{S  \text{ with } k_S=1}  G_{D_N}(z) V_{D_N,D_{N-1}}G_{D_{N-1}}(z) V_{D_{N-1},D_{N-2}} G_{D_{N-2}}(z)\dots V_{D_{2},D_1}
$$
we get the functional equation
\begin{equation}\label{functional_equation}
    G(z) = D(z) + I(z) G(z)
\end{equation}

Note that we have only proven the equality for $z$ with $|\operatorname{Im}(z)|>M$, but $I(z),D(z),G(z)$ are all defined on a larger subsets of $\C$. Now, since we know that $I(z)$ and $D(z)$ are finite combinations of the analytic functions $G_{D_k}(z)$, we get that both are analytic. Hence we may use analytic uniqueness to extend \cref{functional_equation} for other $z \notin \sigma(H^{(N)})$ as well. We do this task in the remaining of this Section. 
\begin{lemma}\label{Domain_of_I,D}
  $I(z)$ and $D(z)$ are well defined for 
  $$
  z \notin \sigma_{cluster}^{(N)}\coloneq\bigcup_{\substack{\text{all non-trivial}\\ \text{ cluster-partitions } \\ D = \{C_i\}_{1 \leq i \le k}}} \sigma (H^{(|C_1|)})+\dots + \sigma (H^{(|C_k|)})
  $$
  with the sum representing the Minkowski sum.
\end{lemma}
\begin{proof}
By a simple inspection note that spectrum of $H_{D_k}$ is equal to the spectrum of $$H^{(|C_1|)} \otimes 1 \dots \otimes 1+\dots + 1 \otimes \dots \otimes H^{(|C_k|)}.$$
This means (cf.  \cite[Corollary 7.25]{schmudgen2012unbounded}) that $\sigma(H_{D_k})=\overline{ \sigma (H^{(|C_1|)})+\dots + \sigma (H^{(|C_k|)})} $ i.e. $G_{D_k}(z)$ is well-defined on the complement of the above set.  Consequently,  $I(z)$ and $D(z)$ are well defined for  $z \notin \overline{\sigma_{\mathrm{cluster}}^{(N)}}$. Hence it suffices to prove that $\sigma (H^{(|C_1|)})+\dots + \sigma (H^{(|C_k|)})$ is closed. \\
To this end, consider the unitary operator $U$ which is given by its action on basis vectors: 
\begin{equation} \label{eq: shift unitary}
    U \ket{m_1,\dots,m_N} = \ket{m_1+1,\dots,m_N+1}.
\end{equation}
By a simple computation one can see that $UH_0^{(N)}U^*=H_0^{(N)} - 2hN$ and $UV^{(N)}U^*=V^{(N)}$. Since $U$ is unitary this means that $\sigma(H^{(|C_i|)})$ is invariant under addition by $2h|C_i|$ and hence it is  periodic with period $2h |C_i|$. We can conclude the  proof by recalling the following elementary fact: 
 If two closed periodic subsets of $\mathbb{R}$ has a common period (not necessarily their minimal period) then their Minkowski sum  is  closed.
\end{proof}

Note that $\sigma_{\mathrm{cluster}}^{(N)} \subseteq \R$. Hence, by analytical uniqueness of holomorphic functions \cite[Theorem 3.10.4]{HillePhillips1948}, we get that \cref{functional_equation} extends to all of $\C \setminus (\sigma(H^{(N)}) \cup \sigma_{\mathrm{cluster}}^{(N)})$.

\begin{lemma}\label{I(z)_is_compact}
    For any $z \notin \sigma_{\mathrm{cluster}}^{(N)}$, we have that $I(z)$ is compact.
\end{lemma}
\begin{proof}
    Similarly to \cite[Lemma 3.1]{article}, we use that if $I(z)$ is a holomorphic function on $U \subset \C$, and if  $I(z)$ is compact on some non-empty open set $V \subset U$, then $I(z)$ is compact for all $z \in U$. See \cite[Lemma 5 on pg. 126]{reed1978iv} for a proof of this fact.\\
    In our case, we take $U = \{z \in \C \mid |\operatorname{Im}(z)| >M\}$. Then
    $$
    I(z) = \sum G_0(z) V_{\alpha_1} G_0(z) \dots G_0(z) V_{\alpha_n} 
    $$
    is norm-convergent, with the sum over all $(\alpha_1,\dots,\alpha_n)$ such that the graph is connected, and  $(\alpha_1,\dots,\alpha_{n-1})$ is not connected. Since the set of compact operators is a closed linear subspace, it suffices to prove that  $G_0(z) V_{\alpha_1} G_0(z) \dots G_0(z) V_{\alpha_n}$ is compact for fixed $(\alpha_1,\dots,\alpha_n)$.\\
    For this, for $R \in \N$ and $\alpha_i = (\alpha_{i,1},\alpha_{i,2}) \in \Lambda^{(N)}$, define $V_{\alpha}^R$ to be the operator with matrix elements
    \begin{align} \label{def: V_R}
    \bra{\underline{n}}V_{\alpha_i}^R\ket{\underline{m}}= \begin{cases}
        \bra{\underline{n}}V_{\alpha_i}\ket{\underline{m}} & \text{if } |m_{\alpha_{i,1}}- m_{\alpha_{i,2}}|\leq R,\\&  \,\,\,\,\,\,\,\,|n_{\alpha_{i,1}}-m_{\alpha_{i,1}}|\leq R, \\&  \,\,\,\,\,\,\,\,\,\,|n_{\alpha_{i,2}}- m_{\alpha_{i,2}}|\leq R\\
        0 & \text{else}
    \end{cases}
    \end{align}
    First, we prove that $\lim_{R\to \infty} V_\alpha^R = V$ in operator norm. Without loss of generality, let us fix $\alpha=(1,2)$. By Schur's test we have: 
    $$\| V_\alpha - V_\alpha^R\|^2 \leq\sup_{m_1,m_2}\sum_{n_1,n_2} |\bra{n_1,n_2}V_\alpha - V_\alpha^R\ket{m_1,m_2}|\sup_{n_1,n_2}
    \sum_{m_1,m_2} |\bra{n_1,n_2}V_\alpha - V_\alpha^R\ket{m_1,m_2}|. $$
     We show that  the first term converges to $0$ when $R\to \infty$, the second is exactly the same. Fix $m_1, m_2$ such that $|m_1-m_2|>R$ then we have:  
     \begin{align} 
    \sum_{n_1,n_2} |\bra{n_1,n_2}V_\alpha - V_\alpha^R\ket{m_1,m_2}| &=  \sum_{n_1,n_2} |\bra{n_1,n_2}V_\alpha \ket{m_1,m_2}| \nonumber \\
    &\leq \sum_{n_1,n_2} \sum_{j_1,j_2} |v(j_1-j_2)| |\mathcal{J}_{n_1-j_1}\mathcal{J}_{m_1-j_1} \mathcal{J}_{n_2-j_2} \mathcal{J}_{m_2-j_2}|\nonumber \\
    &\leq C f(m_1-m_2), \label{eq: |m1-m2|>R}
    \end{align}
    where in the first line we used the definition of $V_{\alpha}-V_{\alpha}^R$, in the second bound we used the definition of $V_{\alpha}$, and in the last bound we performed the sum over $n_1$ and $n_2$ then used the bound \eqref{eq: Bessel Sum} and finally we used \eqref{eq: Bessel decay with V} from Lemma \ref{lem: simple bounds}. \\
    Now fix $m_1,m_2$ such that $|m_1-m_2|\leq R$.  Then thanks to definition of $V_{\alpha}^R$, the sum over $n_1,n_2$ can be decomposed into one sum over $|n_1-m_1|>R$ another sum over and $|n_2-m_2|>R$. We treat the former, the later will be similar: 
    \begin{align}
        \sum_{|n_1-m_1|>R,n_2} |\bra{n_1,n_2} V_{\alpha} \ket{m_1,m_2}| &= \sum \sum_{j_1,j_2} |v(j_1-j_2)| |\mathcal{J}_{n_1-j_1}\mathcal{J}_{m_1-j_1} \mathcal{J}_{n_2-j_2}\mathcal{J}_{m_2-j_2}| \nonumber \\ 
        &\leq C \sum_{r>R}\frac{\exp(cr)}{(r/2)!}\leq (C/R)^{R/2} \label{eq: n1-m1>R}
    \end{align}
    where we used the definition in the first equality as before, then we summed over $n_2$ and used \eqref{eq: Bessel Sum}; then we bounded $|v|$ by its $L^{\infty}$ norm. Then we performed the sum over $j_2$ and used \eqref{eq: Bessel Sum} again. Then, we used \eqref{eq: Bessel decay 1} while summing over $j_1$ for each $n_1$. The sum over $n_1$ with $|n_1-m_1|>R$ is changed into sum over $r$- by paying a constant. Last bound is a simple computation. Treating the case $|n_2-m_2|>R$ and combining this with the case $|m_1-m_2|>R$ finishes the proof of the fact that $V_{\alpha}^R \to V_{\alpha}$ in operator norm.

    Denote with $P_R$ the projection onto $\overline{\operatorname{span}}\{\ket{\underline{m}}\mid |\sum m_i | \leq R\}$. It is clear that $\lim_{R\to \infty} P_R G_0(z) = G_0(z)$ in the operator topology. Hence by the continuity of multiplication we have that
    $$
    \lim_{R\to \infty} P_R G_0(z) V_{\alpha_1}^R G_0(z) \dots G_0(z) V_{\alpha_n}^R = G_0(z) V_{\alpha_1} G_0(z) \dots G_0(z) V_{\alpha_n}.
    $$
    Hence it remains to prove that $L_R = P_R G_0(z) V_{\alpha_1}^R G_0(z) \dots G_0(z) V_{\alpha_n}^R$ is finite rank. \\
    It holds by simple inspection that $L_R\ket{\underline{m}} = 0$ if $|\sum m_i |\geq (2n +1)R$ or if there exists any $\alpha_i \in (\alpha_1,\dots,\alpha_n)$ such that $|m_{\alpha_{i,1}}- m_{\alpha_{i,2}}| \geq nR$. Since $(\alpha_1, \dots ,\alpha_n)$ is a connected graph, it holds that $L_R\ket{\underline{m}} = 0$ if $ |m_k- m_l| \geq n R$ for any $k,l$. It is easy to check that the set of $\underline{m} \in \mathbb{Z}^N$ such that $|\sum m_i |\leq (2n+1)R$, and $|m_i-m_j| \leq nR$ is finite. 
\end{proof}

Note that \cref{functional_equation} implies that if $(1-I(z))$ is invertible, then $G(z) = (1-I(z))^{-1} D(z)$. This is the exact setting of the following theorem:

\begin{theorem}[Analytic Fredholm Theorem]\label{analytic_Fredholm_theorem}
    Assume $f(z)$ is a holomorphic family of compact operators on a connected open set $U \subset \C$. Assuming there exists a $z_0$ such that $(1-f(z_0))$ is invertible, then $(1-f(z))$ is invertible on $V \subset U$ with $U\setminus V$ discrete in $U$. In other words, $(1-f(z))^{-1}$ is a meromorphic function on $U$.
\end{theorem}
\begin{proof}
    See \cite[Appendix]{HunzikerWalter1966Otso}.
\end{proof}
Note that $\|I(z)\| \to 0$ if $|\operatorname{Im}(z)| \to \infty$. Hence there exists a $z_0$ such that $\|I(z_0)\|< 1/2$, in particular $1-I(z_0)$ is invertible.

\begin{proof}[Proof of Theorem \ref{thm: spectral localization}]
   We prove that the spectrum of $H^{(N)}$ is countable. We do this by induction. \\
    The base case is clear, since $\sigma(H^{(1)})= -2h\Z$. Assume the theorem holds for all $j < N$. By induction, it follows that $\sigma_{\mathrm{cluster}}^{(N)}$ is countable. Since $\sigma_{\mathrm{cluster}}^{(N)}$ is closed, it follows that $\C \setminus\sigma_{\mathrm{cluster}}^{(N)}$ is (path)-connected. Since we know that $(1-I(z_0))^{-1}$ exists for $|\operatorname{Im}(z)|$ sufficiently large, by \cref{analytic_Fredholm_theorem} we know that 
$\sigma(H^{(N)}) \subseteq \sigma_{\mathrm{cluster}}^{(N)} \cup D_0$,
with 
$$
D_0 = \{ z \in \C \setminus \sigma_{\mathrm{cluster}}^{(N)}\mid  1 \in \sigma(I(z))\}.
$$
We additionally know that $D_0$ only has accumulation points on $\sigma_{\mathrm{cluster}}^{(N)}$ again thanks to Theorem \ref{analytic_Fredholm_theorem} and Theorem \ref{Domain_of_I,D}. Define $D_\varepsilon = \{z \in D_0\mid \operatorname{dist}(z, \sigma_{\mathrm{cluster}}^{(N)})\geq \varepsilon\}$. By definition of $D_0$ having only accumulation points on $\sigma_{\mathrm{cluster}}^{(N)} $, we know that $D_\varepsilon$ is countable for any $\varepsilon>0$, since an uncountable set always has a limit point. Hence we know that
    $$
    \sigma(H^{(N)}) \subseteq  \sigma_{\mathrm{cluster}}^{(N)}  \cup  \bigcup_{n\in \N} D_{1/n}
    $$
    hence countable. \\
    It follows from the spectral theorem, together that a countably supported Borel measure on $\R$ is atomic \cite[Example 9.1]{schmudgen2012unbounded}, that $\mathcal{H} = \mathcal{H}_{pp}$, which finishes the proof.
\end{proof}
Let us finish this section with the following remark. This remark is not used in the rest of the paper. However, it paints a cleaner picture 
of the spectrum of $H^{(N)}$, and gives a classification of $\sigma_{\text{ess}}(H^{(N)})$.
\begin{remark}
Similarly to \cite[Theorem 3.1]{article}, one proves that
$$
\sigma_{\mathrm{cluster}}^{(N)} = \sigma_{\text{ess}}(H^{(N)}).
$$
 From this, again similar to \cite[Theorem 3.1]{article}, one can proves that this implies that
$$
\sigma_{\mathrm{cluster}}^{(N)} = \bigcup_{j = 1}^{N-1} \sigma(H^{(j)}) + \sigma(H^{(N-j)}).
$$
\end{remark}
\section{Decay in the $a = \sum m_i$ direction} \label{sec: decay M}
Write $\bigotimes \ell^2(\Z^N) = \bigoplus_{a \in \Z} H_a$, with $H_a= \overline{\operatorname{span}}\{\ket{\underline{m}}\mid \sum m_i = a\}$. Write $P_a$ for the projection onto $H_a$.\\
Let $\psi_{\lambda}$ be an eigenvector of $H^{(N)}$, with corresponding eigenvalue $\lambda$. In this section we prove that for any $\theta\geq 0$ we have that
$$
\| P_a \psi_{\lambda}\| < C(\theta) e^{-\theta |a- \lambda|},
$$
for some $C(\theta)$ independent of $\lambda$. \\
Let $R \geq 0$ be any number, to be fixed later. Define 
$$
P_{I} = \sum_{|a-\lambda|<R} P_a
$$
and $P_{T} = \operatorname{Id} - P_{I}$. Note the following equality
\begin{align*}
(H_0- \lambda +V) \psi_{\lambda} = 0 &\implies (H_0 - \lambda) \psi_{\lambda}  = -V \psi_{\lambda}\\
&\implies (H_0 - \lambda) P_T \psi_{\lambda} = -P_T V P_I \psi_{\lambda} - P_T V P_T \psi_{\lambda}\\
&\implies (P_T(H_0 - \lambda)P_T+ P_T V P_T ) P_T \psi_{\lambda} = -P_T V P_I \psi_{\lambda}
\end{align*}
where in the second line we used the fact that $[P_T,H_0]=0$.   
Let $\sigma_{T}(P_T(H_0-\lambda)P_T)$ denotes the spectrum of $P_T(H_0-\lambda)P_T$ as an operator on $\text{range}(P_T)$ then $\operatorname{dist}(0,\sigma_T(P_T(H_0 - \lambda)P_T)) ) \geq R$. Hence, when $R > \|V\|$, we have that $(P_T(H_0 - \lambda)P_T+ P_T VP_T )$ is invertible on $\operatorname{range}(P_T)$. We will later take $R$ to be even larger, but for now we get the final identity 
\begin{equation}\label{eq: Resolvent_identitys}
    P_T  \psi_{\lambda} = - (P_T(H_0 - \lambda)P_T + P_T V P_T)^{-1} (P_T V P_I)  (P_I \psi_{\lambda})
\end{equation}
in particular we get that $\psi_{\lambda}$ is completely determined by $P_I \psi_{\lambda}$.\\
For this, for the remaining of this section, fix $\theta\geq 0$ and $\lambda$. Define the unbounded operator $W = \sum_{a\in \Z} e^{\theta |a - \lambda|} P_a$, with its canonical domain that makes it self-adjoint. For readability, we will not write the dependence of $W$ on $\theta$ and $\lambda$ explicitly.
\begin{lemma}
    Let $\theta \geq 0$, and $\lambda \in \R$. Then $W^{-1} V W$ has bounded closure. We have that $\| \overline{W^{-1} V^{(N)} W} \| \leq C e^{e^{3\theta}}$, with $C$ independent on $\theta$.
\end{lemma}
\begin{proof}
    First, we want to prove that $\|P_aVP_b \|$ decays super-exponentially in $|a-b|$. To this end, it is sufficient to prove this fact for $\| P_aV_{\alpha}P_b\|$, for all interactions $\alpha$. Without loss of generality we may take $\alpha=(1,2)$. Using Schur's test similar to Lemma \ref{I(z)_is_compact}, the proof boils down to finding an appropriate bound for the following: 
    \begin{align*}
        \sup_{\underline{m}} \sum_{\underline{n}} |\bra{\underline{n}} V_{12} \ket{\underline{m}}| =\sup_{m_1,m_2} \sum_{n_1,n_2} |\bra{n_1,n_2}V_{12}\ket{m_1,m_2}| =:\sup_{m_1,m_2} I_{m_1,m_2}
    \end{align*}
    where the first sup is over all $\underline{m}$ such that $\sum_i m_i=a$ and the first sum is over all the $\underline{n}$ such that $\sum n_i =b$. Then the second $\sup$ is over all possible $m_1,m_2$ and the second sum is over all $n_1,n_2$ such that $(m_1+m_2)-(n_1+n_2)=a-b$.\\
    Then, for a given $m_1, m_2$, we have: 
    \begin{align*}
        I_{m_1,m_2} &\leq \sum_{n_1,n_2} \sum_{j_1,j_2} |v(j_1-j_2)| |\mathcal{J}_{n_1-j_1}| |\mathcal{J}_{m_1-j_2}| |\mathcal{J}_{n_2-j_2}| |\mathcal{J}_{m_2-j_2}| \\ &\leq C \sum_{n_1,n_2} \frac{\exp(c|m_1-n_1|)}{|(m_1-n_1)/2|!} \frac{\exp(c|m_2-n_2|)}{|(m_2-n_2)/2|!}  \\
        & \leq C /|(a-b)/6|!
    \end{align*}
    in the first and second line the sum is is over $(n_1+n_2)-(m_1+m_2)=a-b$. In the first bound, we used the definition, in the second bound we performed $L^{\infty}$ over $v$ and then used \eqref{eq: Bessel decay 1} for summations over $j_1$ and $j_2$. In the last bound, we used the fact that $1/|a|!|b|! \leq 
 1/ | a/3|! | b/3|! |(a+b)/3|! $. Then we observe that $1/|(|m_1-n_1|-|m_2-n_2|)/6|! \leq 1/|(m_1+m_2-(n_1+n_2))|=1/|(a-b)/6|!$. Since in the sum this quantity is constant we factor it out and then we sum over all $n_1,n_2$. Finally, we used the fact that $\sum_{n}\exp(c|n-m|)/|(m-n)/C|!  \leq Ce^{cC} \leq C$.
Further, to bound $W^{-1} V W$ note that it defines an infinite matrix, with matrix elements 
$$
\bra{\underline{n}}W^{-1} V W\ket{\underline{m}} = e^{\theta(|b-\lambda| - |a-\lambda|)}\bra{\underline{n}} V \ket{\underline{m}}.
$$
Using Schur's test we prove that this infinite matrix defines a bounded operator.\\
Schur's test then implies that we get that $\|\overline{W^{-1} V W}\| \leq \gamma \beta$, with
    $$
    \gamma = \sup_{b\in \Z}  \sum_{a\in \Z} \|P_a V P_b \|e^{\theta (|a-\lambda| - |b-\lambda|)},\,\,\,\, \beta = \sup_{a \in \Z } \sum_{b\in \Z} \|P_a V P_b \|e^{\theta (|a-\lambda| - |b-\lambda|)}.
    $$
Using the previous decay of $\|P_a VP_b\|$, by the fact that $||a-\lambda| - |b-\lambda|| \leq |a-b|$, we deduce that $\gamma,\beta$ are bounded by
$
C e^{e^{6\theta}}
$by the definition of the exponential function. Hence we get that $\| \overline{W^{-1} V^{(N)} W} \| \leq C e^{e^{3\theta}}$.
\end{proof}

\begin{theorem}
 Let $\theta \geq 0$. Then there exists a $C(\theta)$ such that for any eigenvector $\psi_{\lambda} $ of $H^{(N)}$, with corresponding $ \lambda$ eigenvalue, we have that $ |\langle \psi_{\lambda} \mid m_1,\dots m_N\rangle| \leq C(\theta)e^{-\theta|\sum m_i - \lambda|}$.
\end{theorem}
\begin{proof}
    Set $R =2 \max( \| \overline{W V W^{-1}}\|, \|V\|)$. Define on
    the domain of $W$ the linear functional $\rho$ such that 
    $$
    \rho: D(W) \to \C \mid \xi \mapsto \langle W \xi \mid \psi_{\lambda}\rangle
    $$
    Note that we have that  $\rho(\ket{\underline{m}}) = e^{\theta|\sum m_i - \lambda|}    \langle \underline{m} \mid \psi_{\lambda} \rangle $. We prove that $\rho$ is uniformly bounded, by some $C(\theta)$, independent of $\lambda$.\\
Since $\lambda$ is fixed, we omit the subscript of $\psi_{\lambda}$. Define, for readability,
$$
(H_0-\lambda)_{T,T} = P_T(H_0-\lambda) P_T,\,\,\,\,\,   V_{T,T} = P_TVP_T,\,\,\, V_{I,T} = P_IV P_T,\,\,\, \psi_I = P_I \psi , \,\, \text{etc.}
$$
Hence we get that
$$
\psi_T =  - ((H_0 - \lambda)_{T,T} + V_{T,T})^{-1} V_{T,I}  \psi_I
$$

Note that for $\xi \in D(W)$, since $R > \| W^{-1} V_{T,T} W  \|$ and that $[W,H_0]=0$:
\begin{align}
    W^{-1} (  (H_0 - \lambda )_{T,T}  + V_{T,T} )^{-1} W \xi &= W^{-1}\sum_{n=0}^\infty  (  (H_0 - \lambda  )_{T,T}^{-1}  V_{T,T} )^n  (H_0-\lambda )_{T,T}^{-1} W \xi\\
    &=  \sum_{n=0}^\infty  \left(  (H_0 - \lambda  )_{T,T}^{-1}   (W^{-1}V_{T,T}W ) \right)^n  (H_0-\lambda )_{T,T}^{-1})  \xi\\
    &=  \left(  (H_0 - \lambda )_{T,T}+ \overline{W^{-1} V_{T,T} W} \right)^{-1}  \xi \label{eq: Dualityof}
\end{align}
In particular, we have that 
\begin{equation*}
\|W^{-1} (  (H_0 - \lambda )_{T,T}  + V_{T,T} )^{-1} W\| \leq \frac{1}{R - \|\overline{W^{-1}V_{T,T}W}\|} \leq 2    
\end{equation*}
Indeed, one sees that,
    \begin{align*}
        |\langle W \xi \mid \psi_T\rangle| &= |\langle W\xi \mid((H_0 - \lambda)_{T,T} + V_{T,T})^{-1} V_{T,I}  \psi_I\rangle|\\
        &= |\langle V_{T,I}^*((H_0 - \lambda)_{T,T} + V_{T,T})^{-1} W\xi \mid  \psi_I)\rangle|\\
        &= |\langle W^{-1}V_{I,T}((H_0 - \lambda)_{T,T} + V_{T,T})^{-1} W\xi \mid W \psi_I)\rangle|\\
        &= |\langle \overline{W^{-1}V_{I,T}W}((H_0 - \lambda)_{T,T} + \overline{W^{-1}V_{T,T}W})^{-1} \xi \mid W \psi_I\rangle|\\
        &\leq 2 \|\overline{W^{-1}V_{I,T}W}\|  \| \xi\| \| W \psi_I \|\\
        &\leq  C e^{e^{6\theta}} e^{\theta R} \|\xi\|
    \end{align*}
where we first used \cref{eq: Resolvent_identitys}, in the second line the definition of the adjoint, in the third line that $W$ is self adjoint and restrict to a bounded operator, of norm $e^{\theta R}$, on $\operatorname{range}(P_I)$. Then we used \cref{eq: Dualityof}, after which we used the previous line.\\
    Since we know that $\psi = \psi_T + \psi_I$,  we get that $|\rho(\xi )|\leq (C e^{e^{3\theta}}+1) e^{R\theta}\|\xi\|$ is uniformly bounded, in particular $\rho(\ket{\underline{m}}) = e^{\theta|\sum m_i - \lambda|} |\langle \psi_{\lambda} \mid \underline{m}\rangle| $ is uniformly bounded. Since $R$ was not dependent on $\lambda$, we get that $C(\theta)$ is independent of $\lambda$.
\end{proof} 
From the above theorem we can deduce that:
\begin{corollary} \label{cor: sum decay}
    Let $\theta > 0$, then there is $C(\theta)$ such that for any eigenvalue/vector pair $(\lambda,\psi_{\lambda})$ we have: 
    $$|\langle \psi_{\lambda} \mid e_{j_1},\dots,e_{j_N}\rangle| \leq C(\theta) \exp(-\theta |\sum_ij_i -\lambda|)$$
\end{corollary}
\begin{proof}

Since $\langle \underline{m}, e_{\underline{j}}\rangle = \mathcal{J}_{m_1-j_1} \mathcal{J}_{m_2-j_2} \dots \mathcal{J}_{m_N-j_N}$, this follows by a simple computation using \eqref{eq: Bessel UB}.
\end{proof}

\section{Decay in all $|m_i|$ direction} \label{sec: decay all dir}
In this section, we will prove \cref{thm: superexp loc}. Take $\theta \in \C$ with $\operatorname{Re}(\theta)>0$, and define 
$$
W_{i, \theta} \ket{\underline{m}} =   e^{\theta |m_i| } \ket{\underline{m}} 
$$
as an unbounded, self adjoint operator, on $D(W_{i,\theta})$. We prove \cref{thm: superexp loc} by proving that for an eigenvectors $\psi_\lambda$ that $\psi_\lambda \in D(W_{i,\theta})$. Of course, for us the main question is for $\theta \in \R^+$, but we will need complex $\theta$ do do analytic perturbation theory. \\
For the remainder of the section, fix $i$. Unless specifically required otherwise, we will drop that $i$-dependence, and just write in $W_{\theta}$.\\
Note that $[W_{\theta},H_0] = 0$.  

\begin{lemma}\label{lem: holom 1}
    Fix $1\leq i \leq N$. For every $\theta \in \C$ with $\operatorname{Re}(\theta) >0$ we have that: 
    \begin{enumerate} 
        \item $W^{-1}_{\theta} V W_{\theta}$ has bounded closure,
        \item $W^{-1}_{\theta} V W_{\theta}$ is continuous in $\theta$, that is 
        $$
        \|W^{-1}_{\theta+h} V W_{\theta+h }-W^{-1}_{\theta} V W_{\theta}\| \xrightarrow[|h| \to 0]{} 0
        $$
        \item Let $|m_i|$ denotes the multiplication operator by $|m_i|$, and $[\cdot,\cdot]$ denotes the commutator. Then $W_{\theta}^{-1}[|m_i|,V]W_{\theta}$ has a bounded closure, 
        \item Finally  we have that $W_{\theta}^{-1} V W_{\theta}$ is holomorphic in $\theta$, with derivative $W^{-1}_{\theta}[|m_1|,V_{12}]W_{\theta}$.
    \end{enumerate}
\end{lemma}
\begin{proof}
     We may assume that $\operatorname{Im}(\theta)=0$. In case that $\operatorname{Im}(\theta)\neq 0$, all the above operators will be conjugated with a unitary operator which does not change the norm. Note that, by linearity and similarity, it is sufficient to prove this for $V_{(1,2)}$. Since if $i \notin \{1,2\}$ we have that $W_{\theta}$ commutes with $V_{(1,2)}$. By similarity of $m_1$ and $m_2$ we may assume $i = 1$. Furthermore, we may assume we are in the two-particle case.\\
 Then again appealing to Schur's test, to prove Item 1, it is sufficient to prove that the following expression is bounded: 
 \begin{align}
&\sup_{m_1,m_2}\sum_{n_1,n_2} |\bra{m_1,m_2} W_{1,\theta}^{-1}V_{12}W_{1,\theta}\ket{n_1,n_2}|  =  \nonumber \\
&\sup_{m_1,m_2}
\sum_{n_1,n_2} e^{\theta(|n_1|-|m_1|)} |\bra{m_1,m_2}V_{12}\ket{n_1,n_2}| =: \sup_{m_1,m_2} \sum_{n_1,n_2} I_{m_1,m_2}^{n_1,n_2}(\theta) =: \sup_{m_1,m_2} I_{m_1,m_2}(\theta). \label{eq: I def}
\end{align}
Similarly, to prove Items $2,3,4$ one needs to prove that 
$$
\sup_{m_1,m_2} \sum_{n_1,n_2} I_{m_1,m_2}^{n_1,n_2}(\theta) \varphi_k=: \sup_{m_1,m_2} I^k_{m_1,m_2}(\theta)
$$
is bounded, with
$$
\varphi_k = \begin{cases}
    |e^{h(|n_1| - |m_1|)}-1| & \text{if } k=2\\
    ||n_1|-|m_1|| & \text{if } k = 3\\
    \left|\frac{\exp(h (|n_1|-|m_1|) -1}{h} -( |n_1|-|m_1|)\right| & \text{if } k=4
\end{cases}
$$
where $k$ corresponds to the item number. 
\textit{Item 1}) We observe that $I_{m_1,m_2}$ is bounded uniformly in $m_1,m_2$. Fix $m_1,m_2$, then we have: 
\begin{align}
    I_{m_1,m_2}(\theta) &\leq \sum_{n_1,n_2} \sum_{j_1,j_2} e^{\theta(|n_1|-|m_1|)} |v(j_1-j_2)| \mathcal{J}_{n_1-j_1}| |\mathcal{J}_{m_1-j_1}|
    |\mathcal{J}_{n_2-j_2}||\mathcal{J}_{m_2-j_2}|\nonumber \\
    &\leq  C \sum_{n_1,j_1} e^{\theta(|n_1|-|m_1|)} |\mathcal{J}_{n_1-j_1}||\mathcal{J}_{m_1-j_1}| \leq C \sum_{n_1} \frac{e^{(c+\theta)|n_1-m_1|}}{|(n_1-m_1)/2|!} \leq C(\theta), \label{eq: I bound 1}
\end{align}
uniformly in $m_1,m_2$ where first line is the definition, in the second line we summed over $n_2$, used \eqref{eq: Bessel Sum}, then bounded $V$ by $\|v \|_{\infty}$, and then we summed over $j_2$ and used \eqref{eq: Bessel Sum} again. In the last bound we used \eqref{eq: Bessel decay 1} upon summing over $j_1$, and we used the fact that $||n_1|-|m_1|| \leq |n_1-m_1||$. The last bound is deduced from Taylor expansion of the exponential function (it can be bounded by $C(e^{c+\theta}+1)e^{e^{c+\theta}}$). \\
\textit{Item 2}) Now for Item 2, we have $\varphi_2=
    |\exp(h(|n_1|-|m_1|))-1|$. We use the basic bounds $|e^t-1| \leq |t|e^{|t|}$ for any $t \in \mathbb{C}$, and $|x| \times e^{|hx|} \leq C(\theta)(\exp(\theta x/2)+\exp(-\theta x/2)) $ for any $x \in \mathbb{R}$ and $|h|$ sufficiently small. From these two bounds we deduce that $I^{2}_{m_1,m_2}(\theta) \leq |h|C(\theta)(I_{m_1m_2}(\theta/2)+I_{m_1,m_2}(3\theta/2))$. Combining this with \eqref{eq: I bound 1}, namely the fact that for any $\theta$ that $I_{m_1,m_2}(\theta)$ is bounded uniformly in $m_1,m_2$ by some constant $C(\theta)$, finishes the proof of Item 2.\\
\textit{Item 3}) notice that $\varphi_3= ||n_1|-|m_1||$. In this case using the general bound
    $|t| \leq C(\theta)(\exp(\theta t/2)+\exp(-\theta t/2))$ for any $t \in \mathbb{R}$, we deduce that $I^3_{m_1,m_2}(\theta) \leq C(\theta)(I_{m_1,m_2}(\theta/2)+I_{m_1,m_2}(3\theta/2))$ which finishes the proof thanks to \eqref{eq: I bound 1} (similar to Item 2).\\
    \textit{Item 4}) We observe that $\varphi_4=|\frac{\exp(h(|n_1|-|m_1|))-1)}{h}-(|n_1|-|m_1|)|$. In this case, we use the following general facts that for any $x \in \mathbb{R}$ and $h \in \mathbb{C}$ with $|h|$ sufficiently small: 
    $$\left|(e^{hx}-1)/h -x \right| \leq \frac{|h|x^2}{2}e^{|h||x|} \leq |h| C(\theta)  (e^{\theta x/2}+e^{-\theta x/2}).$$ 
    Having the above bound, we observe that $I_{m_1,m_2}^4 \leq |h|C(\theta) (I_{m_1,m_2}(\theta/2)+I_{m_1,m_2}(3\theta/2))$. Similar to Item 1, this finishes the proof of Item 4 appealing to \eqref{eq: I bound 1}.
\end{proof}

When we write $W_\theta^{-1} H^{(N)} W_\theta$, we will mean the operator $ H_0^{(N)} + W_\theta^{-1} V W_\theta$ with domain $D(H_0)$. Of course, both agree on $D(H_0^{(N)}) \cap D(W_\theta^{-1} H^{(N)} W_\theta)$. The latter has larger domain, in particular it is a closed operator.\\

%


    
\begin{lemma}\label{compact_semisimilar}
    Let $K$ be a compact operator on a Hilbert space $\mathcal{H}$. Let $W: D(W) \to \mathcal H$ be an bijective closed operator. So we know that $W^{-1}$ is bounded and injective with dense range. Assume that $W^{-1}KW$ has a compact bounded closure $B$. Then $\sigma(B) = \sigma(K)$. 
\end{lemma}
\begin{proof}
    We know that $W^{-1} K = B W^{-1}$ as bounded operators. Since $0 \in \sigma(B)$ and $0 \in \sigma(K)$ for all compact operators, we may look at $z \in \sigma(K) \setminus \{0\}$. By the theory of compact operators, we know that $z$ is an eigenvalue. Hence let $\phi$ such that $K \phi = z\phi$. By injectivity we know that $W^{-1} \phi \neq 0$. Then one checks that 
    $$
    B (W^{-1} \phi) = W^{-1} K \phi = z(W^{-1}\phi) 
    $$
    hence $z \in \sigma(B)$. Similarly, if $z \in \sigma(B)\setminus\{0\}$, we have that $\overline{z} \in \sigma(B^*)\setminus\{0\}$. Since $B^*$ is again compact, one finds an eigenvector $\phi$ such that $B^* \phi = \overline{z}\phi$. One checks that $(W^{-1})^*$ is again injective with dense range, hence $(W^{-1})^*\phi$ is non-zero. One sees that
    $$
    K^* ((W^{-1})^*\phi) = (W^{-1})^* B^*\phi = \overline{z} ((W^{-1})^*\phi)
    $$
    so $\overline{z} \in \sigma(K^*)$, hence $z \in \sigma(K)$.\\
\end{proof}
\begin{remark}
    The next proof of the next theorem will make extensive use of the equality $W_{\theta}^{-1} (z-A)^{-1} W_{\theta} = (z- W_\theta^{-1} A W_\theta)^{-1}$, for different operators $A$. This equality can be checked to holds on $D(W_\theta)$, by checking that the left hand side is indeed an inverse of the right hand side. Since $D(W_\theta)$ is dense, the following equivalence follows: $\sigma(A) = \sigma(W_{\theta}^{-1} A W_\theta) $ if and only if $W_{\theta}^{-1} (z-A)^{-1} W_{\theta}$ has closure for $z \in \C \setminus \sigma(A)$.
\end{remark}

\begin{theorem} \label{lem:technical W}
Fix $1\leq i \leq N$. Let $\theta \in \C$ with $\operatorname{Re}(\theta) \geq 0$. Then the following are true: 
\begin{enumerate}
        \item $W_{\theta}^{-1}G_{D_k}(z)W_{\theta}$ has a bounded extension for all non-trivial cluster decompositions $D_k = \{C_i\}_{1\leq i\leq k}$ and $z \in  \C \setminus (\sigma(H^{|C_1|}) + \sigma(H^{|C_2|}) +  \dots + \sigma(H^{|C_k|}))$,
        \item $W_{\theta}^{-1}D(z)W_{\theta}$ and $ W_{\theta}^{-1} I(z)W_{\theta}$ have a bounded extension for all $z \in  \C \setminus \sigma_{\mathrm{cluster}}^{(N)}$,
        \item $W_{\theta}^{-1} I(z)W_{\theta}$ is compact for every $z\in \Dom(I) = \C \setminus \sigma_{\mathrm{cluster}}^{(N)}$,
        \item $W_{\theta}^{-1}G(z)W_{\theta}$ has a bounded extension for all $z \in \Dom(G) = \C \setminus \sigma(H^{(N)})$.
    \end{enumerate}
\end{theorem}
\begin{proof}
    We prove this theorem by induction on $N$.The base step of $N=2$ is easy. We proceed to the induction step. Up until Item 4, we may assume $\theta$ real, since the complex part corresponds to a unitary.\\ 
    a) One sees that $W_{\theta}^{-1}G_{D_k}(z)W_{\theta} = (z- W^{-1}_{\theta} H_{D_k} W_{\theta} )^{-1}$ on the domain $W_{\theta}$. Hence it suffices to prove that $\sigma(H_{D_k}) = \sigma(W^{-1}_{\theta} H_{D_k} W_{\theta})$. Since the resolvent of $W^{-1}H_{D_k}W$ is bounded and well defined on the complement of the spectrum.\\
    Since $D_k$ is non-trivial, there is a $j$ with $1 \leq j <N$ such that  we may write $H_{D_k} = 1 \otimes H^{(j)}_{D_{k_1}} + H_{D_{k_2}}^{(N-j)}\otimes 1$. Depending on whether $i\leq j$ or $i>j$ ($i$ appeared in the definition of $W_{\theta}\equiv W_{i, \theta}$) $W_{\theta}$  acts non-trivially on only one of the above terms in the sum.  Without loss of generality, assume $i\leq j$ i.e. 
    $$ 
        W_{\theta}^{-1}H_{D_k}W_{\theta} = 1 \otimes (W_{\theta}^{-1}H_{D_{k_1}}^{(j)}W_{\theta}) +  H_{D_{k_2}}^{(N-j)} \otimes 1.
    $$
    Then by induction, we know by a) or by d), that $\sigma(W_{\theta}^{-1}H^{(j)}_{D_{k_1}}W_{\theta}) = \sigma(H^{(j)}_{D_{k_1}})$. We note that for unbounded (non-selfadjoint) operators $A,B$, 
    it is not true that $\sigma(1 \otimes A + B \otimes 1) = \overline{\sigma(A) + \sigma(B)}$, for further discussion see \cite{Ichinose+1970+119+153, REED1973107}. We prove it for the case that one of them is selfadjoint.\\
    Take $z \notin \sigma(H_{D_k})$, hence $z \notin \sigma(H_{D_{k_1}}^{(j)}) + \sigma(H^{(N-j)}_{D_{k_2}})$. Then we write 
    $$
    A(z) = \sum_{\lambda \in \sigma(H^{(N-j)}_{D_{k_2}})}  P_{\lambda}\otimes ((z-\lambda) - W_{\theta}^{-1}H^{(j)}_{D_{k_1}}W_{\theta})^{-1 }, 
    $$
    where $P_{\lambda}$ is the spectral projection to the subspace corresponding to $\lambda$.
    It is easy to check that if $A(z)$ were to converge, then it would be an inverse of $z-W_{\theta}^{-1}H_{D_k}W_{\theta} $. So it rests us to check convergence.\\
    Since $P_{\lambda}$ are orthonormal projections, it suffices to prove that $\|((z-\lambda) - W_{\theta}^{-1}H^{(j)}_{D_{k_1}}W_{\theta})^{-1 }\|$ is uniformly bounded, independent of $\lambda$. \\
    This follows since $\operatorname{dist}(z-\lambda, \sigma(W_{\theta}^{-1} H_{D_{k_1}}^{(j)} W_{\theta})) \geq d>0$. Additionally, $\|(z- W_{\theta}^{-1} H_{D_{k_1}}^{(j)} W_{\theta})^{-1}\|$ is invariant under the transformation $z \mapsto z+ j\times2h$, and converges to $0$ when $|\operatorname{Im}(z)| \to \infty$. In particular,we get that $\|(z- W_{\theta}^{-1} H_{D_{k_1}}^{(j)} W_{\theta})^{-1}\|$ is uniformly bounded in $\lambda \in \sigma(H^{(N-j)}_{D_{k_2}})$, since  
    $$
    \{ z\in \C | |\operatorname{Im}(z)|\leq M,\operatorname{dist}(z,\sigma(H^{(j)}_{D_{k_1}}) \geq d, 0 \leq\operatorname{Re}(z)\leq j\times h\}
    $$ is compact.\\
    b). Since 
    $$
    D(z)  = \sum_{S  \text{ with } k_S\geq 2} G_{D_N} V_{D_N,D_{N-1}}G_{D_{N-1}} V_{D_{N-1},D_{N-2}} G_{D_{N-2}}\dots V_{D_{k+1},D_k}G_{D_k}
    $$
    and 
    $$
    I(z) = \sum_{S  \text{ with } k_S=1}  G_{D_N}(z) V_{D_N,D_{N-1}}G_{D_{N-1}}(z) V_{D_{N-1},D_{N-2}} G_{D_{N-2}}(z)\dots V_{D_{2},D_1}
    $$
    with $D_k$ never being the trivial cluster $(1,2,\dots,N)$. Hence this follows directly from a).\\
    c) Take $\operatorname{Im}(z) > \frac{1}{{N\choose 2}\|W_{\theta}^{-1} V_\alpha W_{\theta}\|}$. Then we have, similar to the proof of \cref{I(z)_is_compact}, together with the fact that $[G_0(z),W] = 0$, that 
    $$
    W_{\theta}^{-1}I(z) W_{\theta} = \sum_{(\alpha_1,\alpha_2,\dots \alpha_n)} G_0 (W_{\theta}^{-1}V_{\alpha_1}W_{\theta}) G_0 (W_{\theta}^{-1}V_{\alpha_2}W_{\theta})\dots (W_{\theta}^{-1}V_{\alpha_n}W_{\theta} )
    $$
    with the sum over all graphs that are connected, but such that $(\alpha_1,\alpha_2,\dots \alpha_{n-1})$ is not connected. \\
    The proof of the fact that $W_{\theta}^{-1}I(z)W_{\theta}$ is compact is very similar to Theorem \ref{I(z)_is_compact}, mainly since matrix elements of $V_{\alpha_i}$ have super-exponential decay and adding an exponential factor does not change the decay properties. Given the similarity with proof of Theorem \ref{I(z)_is_compact}, we only mention the general steps and main difference and omit the rest for the sake of brevity. \\
    As before, it is sufficient to prove that $L_{W}:=G_0 (W_{\theta}^{-1}V_{\alpha_1}W_{\theta}) G_0 (W_{\theta}^{-1}V_{\alpha_2}W_{\theta})\dots (W_{\theta}^{-1}V_{\alpha_n}W_{\theta} )$ is compact for a given sequence of $(\alpha_1,\dots,\alpha_n)$ appearing in the above sum. We may define $V_{\alpha_i}^R$ as in \eqref{def: V_R} and $P_R$ is also defined as before. Then in the same manner (since $[W_{\theta}^{-1},G_0]=0$) we can see that $L^R_{W}=P_RG_0(W_{\theta}^{-1}V^R_{\alpha_1}W_{\theta}) \dots G_0(W_{\theta}^{-1}V_{\alpha_n}W_{\theta})$ is also finite rank. The only remaining ingredient of the proof is the fact that $W_{\theta}^{-1}V_{\alpha_i}^RW_{\theta}$  converges in norm to $W_{\theta}^{-1}V_{\alpha_i}W_{\theta}$ as $R \to \infty$. At this point, it should not be a surprise to the reader that we will use Schur's test again for this.  Similar to Theorem \ref{I(z)_is_compact} the Schur's test boils down to showing that following expressions converges to zero uniformly in $m_1,m_2$ as $R \to \infty$: 
    $$
    \sum_{n_1,n_2}\sum_{j_1,j_2} |v(j_1-j_2)||\mathcal{J}_{n_1-j_1}\mathcal{J}_{m_1-j_1} \mathcal{J}_{n_2-j_2}\mathcal{J}_{m_2-j_2}| e^{\theta(|m_1|-|n_1|)}
    $$
    In three cases: either when $|m_1-m_2|>R$, or when the sum is over $n_1$ such that $|n_1-m_1|>R$, lastly when the sum is over $n_2$ such that $|n_2-m_2|>R$.\\
    Let us briefly mention why first two cases are identical to \eqref{eq: |m1-m2|>R} and \eqref{eq: n1-m1>R}. Notice that in mentioned bounds the argument of the Bessel function $x$ is omitted and constant in these bounds  depend on $x$.  Moreover,  to obtain \eqref{eq: |m1-m2|>R} and \eqref{eq: n1-m1>R} we first bound Bessel functions via \eqref{eq: Bessel UB} i.e. $|\mathcal{J}_{n}(x)|\leq |x/2|^{|n|}/|n|!$. Now it is sufficient to bound 
    $$e^{\theta(|m_1|-|n_1|)}|x|^{|n_1-j_1|} |x|^{|m_1-j_1|} \leq (|x|e^{\theta})^{|n_1-j_1|} (|x|e^{\theta})^{|m_1-j_1|}. $$
    This bound allow us to get the exact same bound as in \eqref{eq: n1-m1>R} and \eqref{eq: |m1-m2|>R} up to modifying the constants $C(|x|) , c(|x|)$ to $C(|x|e^{\theta}), c(|x|e^{\theta})$. \\
    Finally, notice that although at first sight the term where the first sum is over $|n_2-m_2|>R$ seems different than the previous case (due to the asymmetric exponential factor only depending on $n_1,m_1$), by a simple inspection it is direct that this term can also be treated with the above trick of modifying the constant and argument of the Bessel function. 

    d) We first prove that
    $$
    \sigma(W_{\theta}^{-1}H^{(N)}W_{\theta})\subset \sigma_{\mathrm{cluster}}^{(N)}\cup \{z \in \C\setminus \sigma_{\mathrm{cluster}}^{(N)}\mid 1 \in \sigma(I(z))\}.
    $$
    Indeed, by \cref{functional_equation} we see that
    $$
    (1-W_{\theta}^{-1}H^{(N)} W_{\theta})^{-1} = \left(1- W_{\theta}^{-1} I(z) W_{\theta}\right) ^{-1} (W_{\theta}^{-1} D(z)W_{\theta})
    $$
    if the terms are defined. By b). We know that $W_{\theta}^{-1} D(z)W_{\theta}$ and $W_{\theta}^{-1} I(z) W_{\theta}$ have a bounded extension for $z \notin \sigma_{\mathrm{cluster}}^{(N)}$. So it rests us to prove that $1 \in \sigma(W_{\theta}^{-1} I(z) W_{\theta}) )$ if and only if $1 \in \sigma(I(z))$. But this directly follows from \cref{compact_semisimilar}.\\
    It remains to prove the following: if $z_0 \in \sigma_{\mathrm{cluster}}^{(N)} \cup \{z \in \C\setminus \sigma_{\mathrm{cluster}}^{(N)}\mid 1 \in \sigma(I(z))\}$, and $G(z_0)$ is defined. We prove that $W_{\theta}^{-1} G(z_0) W_{\theta}$ is also well defined. Indeed, since we know that $z_0$ is separated from $ \{\sigma_{\mathrm{cluster}}^{(N)}\cup \{z \in \C\setminus \sigma_{\mathrm{cluster}}^{(N)}\mid 1 \in \sigma(I(z))\}\} \setminus \{z_0\}$, say with distance $d$. Hence we can find a path $\Gamma$ enclosing only $z_0$. Then the Laurent series 
    $$
    \sum_{k = -\infty}^\infty A_k(\theta)  (z-z_0)^k
    $$
    converges in norm on compacts for 
    $$
    0< |z-z_0| < d
    $$
    to $W^{-1}_{\theta} G(z) W_{\theta}$, with 
    $$
    A_k(\theta) = \int_{\Gamma} \frac{W^{-1}_{\theta}G(z) W_{\theta}}{(z-z_0)^{k+1}}dz.
    $$
    We specifically are interested in how $A_k(\theta)$ changes with $\theta$.\\
    Notice that for $z \in \C \setminus \sigma_{\mathrm{cluster}}^{(N)}$ we have that $W_\theta^{-1}G(z) W_\theta$ is holomorphic  for $\operatorname{Re}(\theta) > 0$.  We postpone the proof of this fact to keep the flow of the proof. \\
    Define for $\operatorname{Re}(\theta) < 0$ then $W_{\theta}^{-1}G(z) W_{\theta} \coloneq  (W_{-\theta}^{-1}G(z)^* W_{-\theta})^*$. With this definition, $W_{\theta}^{-1}G(z) W_{\theta}$ is a holomorphic function away from the imaginary axis, with a continuous extension to the imaginary axis. By the symmetry principle for holomorphic functions, we get that it is holomorphic everywhere.\\
    Using this we get that $A_k(\theta)$ is defined and analytic for every $\theta\in \C$. Furthermore, for $k <0$ and $\theta$ such that $||V - W_\theta^{-1}VW_\theta|| < d$, we know that $A_k(\theta) = 0$ by the von Neumann series. Hence $A_k(\theta)$ is an analytic function which is zero inside some open set, hence it is zero everywhere. Hence $W_{\theta}^{-1}G(z) W_{\theta}$ is well defined, meaning that $z_0 \notin \sigma(W^{-1}_\theta H^{(N)}W_\theta)$.
\end{proof}

\begin{proof}[Proof of the fact that $W_\theta^{-1}G(z)W_{\theta}$ is holomorphic in $\theta$]
Fix $z$ as before, let $\theta \in \mathbb{C}$ with $\text{Re}(\theta)>0$. For ease of notation, we define for any $h \in \mathbb{C}$ (sufficiently small in norm):  
\begin{equation} \label{eq: def Delta}
    \tilde{\Delta}(h):= W_{\theta+h}^{-1} VW_{\theta+h} -W_{\theta}^{-1}VW_{\theta}.
\end{equation}
Recall that Lemma \ref{lem: holom 1} tells us that $\lim_{h \to 0}\tilde{\Delta}(h) = 0$ and $\lim_{h \to 0}\frac{\tilde{\Delta}(h)}{h} = W_{\theta}^{-1} [|m_1|,V] W_\theta$. \\
As before, observe that on the domain of $W_{\theta}$ we have: 
    $W_{\theta}^{-1}G(z)W_{\theta}= (z-W_{\theta}^{-1}HW_{\theta})^{-1}=(z-H_0-W_{\theta}^{-1}VW_{\theta})^{-1}=:G_{\theta}$. 
We first prove that $G_\theta$ is continuous in $\theta$, then we prove that it is holomorphic. For continuity, by the von Neumann series, we may write for $|h|$ sufficiently small, and Item 2 of Lemma \ref{lem: holom 1} that 
$G_{\theta + h}  = G_\theta \sum_{n=0}^\infty (\tilde{\Delta}(h) G_\theta)^n $. Hence $\lim_{h\to 0} G_{\theta+h} = G_\theta$, since 
$$
\| G_{\theta+ h} - G_\theta\| \leq \frac{ \|\tilde \Delta(h)\|\|G_\theta\|}{1- \|G_\theta \tilde{\Delta}(h)|} \xrightarrow[{|h| \to 0}]{}  0.
$$
We now finally prove that $G_\theta$ is holomorphic. Note that
\begin{align*}
    \frac{ G_{\theta+h}-G_{\theta}}{h} &=  G_{\theta+h} \frac{\tilde{\Delta}(h)}{h }G_\theta \\
    & \xrightarrow[|h|\to 0]{} G_\theta ( W_{\theta}^{-1}[|m_1|,V] W_{\theta}) G_{\theta}  
 \end{align*}
where in the second line we used  the resolvent identity to rewrite $G_{\theta+h}-G_{\theta}$. Finally the limit as $|h| \to 0$ is given thanks to Item 4 of Lemma \ref{lem: holom 1}, together with the fact that multiplication is continuous in the operator norm. 
\end{proof}
\begin{theorem}\label{bounded_extention_for_projection}
    Let $\mathcal{H}$ be a Hilbert space. Let $W: D(W) \to \mathcal{H}$ be an unbounded self-adjoint bijective operator with dense domain $D(W)$. Then we know that $W^{-1}$ exists and is bounded. Let $P$ be a projection, and assume that $W^{-1}P W$ has a bounded closure. Then $\operatorname{range}(P) \cap D(W) $ is dense as a subset of $\operatorname{range}(P)$.
\end{theorem}
\begin{proof}
    Let $B$ be the bounded extension of $W^{-1}P W$. Fix $y \in H$. Note that 
$$
\sup_{x \in D(W)} |\langle  PW^{-1}y,  Wx\rangle| \leq \|By\| \|x\|
$$
In particular that $PW^{-1} y \in D(W^*) = D(W)$. This is true for any $y \in H$. Since we know that $\operatorname{range}(W^{-1}) = D(W)$, we get that $PD(W) \subset D(W)$.\\
Now since $P$ is a continuous map, with closed range, we get that $PD(W) \subset P \mathcal{H}$ is a dense set. Hence $P D(W) \subset \operatorname{range}(P) \cap D(W)$, and we are done. 
\end{proof}


\begin{theorem}\label{theorem}
    Let $\lambda$ be an isolated eigenvalue of $H^{(N)}$, and denote with $P_\lambda$ the spectral projection. Fix $1 \leq i \leq N$, and let $\theta \in \C$ with $\operatorname{Re}(\theta)\geq 0$. Then  $W_\theta^{-1} P_\lambda W_\theta$ has bounded closure.
\end{theorem}
\begin{proof}
    We have proven that $\sigma(W_\theta^{-1} H^{(N)} W_\theta) = \sigma(H^{(N)})$ thanks to Lemma \ref{lem:technical W} Item 4. In particular, $\lambda$ is also an isolated point of $\sigma(W_\theta^{-1} H^{(N)} W_\theta)
    $. Hence there exists a closed curve $\Gamma$ enclosing $\lambda$, and only $\lambda$. Define
    $$
    B = \int_{\Gamma} (z - W_\theta^{-1} H^{(N)} W_\theta)^{-1} dz
    $$ 
    which is bounded since $\Gamma$ is compact. For $\xi \in D(W_\theta)$ we get that 
    $$
   B\xi  = \int_{\Gamma} (z - W_\theta^{-1} H^{(N)} W_\theta)^{-1} \xi dz  = \int_{\Gamma} W_\theta^{-1}(z -  H^{(N)})^{-1} W_\theta\xi dz  = W_\theta^{-1} P_\lambda W_\theta \xi
    $$
    hence $B$ is a bounded extension of $W_\theta^{-1} P_\lambda W_\theta$.
\end{proof}

\begin{proof}[Proof of Theorem \ref{thm: superexp loc}]
First, we observe that: $P_\lambda$  is finite dimensional. Similar to \cite{HunzikerWalter1966Otso}, notice that from \cref{functional_equation} we have that 
$$
P_\lambda = \lim_{z \to \lambda}(z-\lambda) G(z) = \lim_{z \to \lambda}((z-\lambda)(D(z)+I(z)  G(z))) =  I(\lambda) \lim_{z \to \lambda}(z-\lambda) G(z) = I(\lambda)P_\lambda ,
$$ 
where we used the spectral theorem in the first equality, then we used theorem \ref{Domain_of_I,D} to treat $D(\lambda), I(\lambda)$ and the fact that $\lambda \notin \sigma_{\text{cluster}}^{(N)}$. Finally, we recall the fact that the product of compact and bounded operators is compact and
 $I(\lambda)$ is compact by \cref{I(z)_is_compact}. So we get that $P_\lambda$ is a compact projection, hence finite dimensional.\\
By \cref{theorem} and \cref{bounded_extention_for_projection}, we know that $D(W_\theta)$ is dense in $\text{range}(P_{\lambda})$, being a finite dimensional space this means $\psi_\lambda\in D(W_\theta)$ for any $\theta \in \R$ and $1 \leq i\leq N$. Hence we conclude that for any $\underline{m} \in \mathbb{Z}^{N}$: $$|\langle \psi_{\lambda}\mid \underline{m}\rangle |\leq C(\theta,\lambda) e^{-\theta |m_i|}.$$ Hence $$|\langle \psi_{\lambda}\mid\underline{m}\rangle |^N<  C(\theta,\lambda)  e^{-\theta (|m_1| + |m_2| + \dots + |m_N|)}.$$ Since $\theta$ was arbitrary, and by replacing $\theta'= N\theta$, we conclude the desired decay for the Stark basis. To get the  bound in the position basis, it is sufficient to recall that 
$\langle m_i \mid e_{j_i}\rangle=\mathcal{J}_{m_i-j_i}$. The decay stated in Theorem \ref{thm: superexp loc} can be deduced by a simple computation from the above bound and bounds \eqref{eq: Bessel UB}, and \eqref{eq: Bessel Sum} from Lemma \ref{lem: simple bounds}.
\end{proof}
\appendix

\bibliographystyle{plainnat}
\bibliography{references}
\end{document}